\documentclass{amsart}
\usepackage[english]{babel}
\usepackage{amsrefs}
\usepackage{amssymb,amsmath,amsthm,amsfonts,epsfig,graphicx,psfrag,color}
\usepackage[utf8]{inputenc}
\usepackage{amsrefs}
\usepackage{hyperref}
\hypersetup{
    colorlinks,
    citecolor=black,
    filecolor=blue,
    linkcolor=black,
    urlcolor=blue
}

\theoremstyle{plain}
\newtheorem{teo}{Theorem}[section]
\newtheorem{thm}[teo]{Theorem}
\newtheorem{cor}[teo]{Corollary}
\newtheorem{lem}[teo]{Lemma}
\newtheorem{prop}[teo]{Proposition}		

\theoremstyle{definition}

\newtheorem{rmk}[teo]{Remark}

\DeclareMathOperator{\card}{card}
\DeclareMathOperator{\diam}{diam}

\newcommand{\Spo}{{P_0}}

\newcommand{\R}{\mathbb R}

\renewcommand{\epsilon}{\varepsilon}

\begin{document}

\author[A. Artigue]{Alfonso Artigue}
\address{Departamento de Matem\'atica y Estad\'\i stica del Litoral\\
Universidad de la Rep\'ublica\\
Gral. Rivera 1350, Salto, Uruguay.}

\email{artigue@unorte.edu.uy}

\title{Billiards and toy gravitons}
\date{\today}
\begin{abstract} 
In this article we study the one-dimensional dynamics of elastic collisions of particles 
with positive and negative mass. 
We show that such systems are equivalent to billiards induced by an inner product of possibly indefinite signature, 
we characterize the systems with finitely many collisions and  
we prove that a small particle of negative mass between two particles of positive mass 
acts like an attracting particle with discrete acceleration (at the collisions) 
provided that the total kinetic energy is negative. 
In the limit of the negative mass going to zero, with fixed negative kinetic energy, we obtain a continuous acceleration 
with potential energy of the form $U(r)=-k/r^2$.
\end{abstract}

\maketitle

\section{Introduction}

The dynamics of one ball in a triangular billiard table is known to be equivalent to 
the one-dimensional motion of
three particles on a circle with elastic collisions.
This result, that was proved by Glashow and Mittag \cite{GlMi},
states that obtuse triangles correspond to three particles with masses of different sign. 
Indeed, three particles with masses $m_1,m_2,m_3$ give rise to a triangular billiard if and only if
\begin{equation}
 \label{ecuGM}
 (m_1+m_2+m_3)m_1m_2m_3>0. 
\end{equation}
Under this condition, if the masses have the same sign then the associated triangle is acute, otherwise it is obtuse.
Notice that \eqref{ecuGM} is satisfied if, for instance, 
$m_1$ and $m_3$ are positive and $m_2$ is negative with large absolute value.
The dynamics of triangular billiards is far from trivial.
For example, it is not known whether every obtuse triangular billiard table has periodic trajectories or not. 
The answer is known to be affirmative in some cases, we suggest \cite{Sch08} for relevant results on the triangular billiard problem and
\cite{ChMa} for the general theory of billiards.

Therefore, the obtuse triangular tables, which corresponds to the difficult case of the billiard problem, 
lead us to consider the one-dimensional dynamics of particles with positive and negative mass. 
From a mathematical viewpoint, the relevance of studying the dynamics of negative masses is supported by the obtuse triangular billiard problem.
From a physical viewpoint it is less clear which role negative mass can play, for instance, the mass of the known physical particles 
is zero or positive. 
The literature concerning negative mass, known to the author, is mainly dedicated to gravity, 
see for example \cites{Lut51,Bon57,Bonnor89} (just to name some classics). 
It is possible that we never find negative mass in our universe, 
but we all deal with negative energy, namely, the gravitational potential energy.
In Theorem \ref{thmGravLim} we will show that a negative potential energy of the form $U(r)=-k/r^2$ is the limit kinetic 
energy of a particle of negative mass.

We prove that if a system does not satisfy \eqref{ecuGM} then it is equivalent to a billiard 
defined by an indefinite inner product. 
The \emph{mirror law} for this case is explained in \S\ref{secGeomModel}.
For instance, if we consider a tiny particle $m_2<0$ (tiny in absolute value) 
between $m_1,m_3>0$
we obtain 
what in this paper will be called a \emph{toy graviton}. 
In \S\ref{secColSys} it will be proved that if the velocity of $m_2$ is sufficiently large (giving negative total kinetic energy), then the one-dimensional 
system of three particles will collapse in finite time at the center of mass. 
In some sense, it is a discrete model of an attracting force, where the particles are accelerated only at collisions.

If we wish a continuous attraction, we can take the limit $m_2\to 0$ but with a fixed negative kinetic energy. 
This implies that its velocity diverges. Therefore, in the limit, 
the toy graviton has not a definite position and it is transformed into \emph{pure negative kinetic energy}, but without mass. 
As we said, in our one-dimensional model we obtain a potential energy of the form 
$U=-k/r^2$, which obviously differs from Newtonian potential in the exponent of $r$. 
We remark that in the paper we only consider classical mechanics. It would be interesting to 
know whether considering relativistic collisions or quantum mechanical particles in a higher dimensional space 
could give rise to a realistic
model of gravity.

Let us explain the contents of this article.
In \S\ref{secPartLine} we define the dynamics, we study the signature of the kinetic energy 
and some properties of collisions between particles of masses of different sign are explained. 
In \S\ref{secGeomModel} we consider billiards defined by a non-necessarily positive inner product. 
We show the equivalence between one-dimensional motions of a systems of particles and an associated billiard. 
The evolution of the moment of inertia of the system is studied.
In \S\ref{secColSys} we 
prove that a system has collapsing solutions if and only if the kinetic energy is indefinite. 
We say that a solution collapses if it has infinitely many collisions in finite time. For this 
purpose we introduce two key systems: toy gravitons \S\ref{secToyG} and compressors \S\ref{secCompr}. 
We show that for a system of $N\geq 3$ particles on the line the following statements are equivalent: 
\begin{itemize}
 \item the kinetic energy form restricted to configurations with vanishing momentum has a definite sign,
 \item every solution has finitely many collisions and the total mass is not zero, 
 \item all the masses have the same sign or there is exactly one mass with the sign of the total mass.
\end{itemize}
The first equivalence is given in Corollary \ref{corSgBrS1}. 
The second equivalence is proved in Theorem \ref{thmSgBrS}, it depends on the study of the dynamics of toy gravitos and compressors.
Finally, in \S\ref{secGravity} we prove Theorem \ref{thmGravLim}, that was mentioned above.

\section{Particles on a line}

We start defining the dynamical system that will be consider in this paper. 
Also, we recall some fact concerning quadratic forms and give an elementary result about maximal solutions.

\label{secPartLine}
\subsection{Preliminaries} 
Consider $N$ particles in the real line $\R$. 
The positions and the masses of these particles will be denoted as 
$x_i(t),m_i\in\R$, $m_i\neq 0$, for all $i=1,\dots,N$.
Some restrictions will be imposed on the values of the masses 
but we allow them to be positive and negative. 
Define the total mass as 
$$M=m_1+\dots+m_N$$ 
and denote by $v_i(t)$ the velocity of the $i$-particle. 
We will assume the conservation of momentum and kinetic energy, so let us explain some basic properties of these quantities.

The momentum of a solution is defined as $P=\sum_{i=1}^ Nm_iv_i$. 
If $M\neq 0$ then we can define new coordinates $y_i=x_i-t P/M$. 
With respect to these coordinates the system has the center of mass 
$\sum_{i=1}^ N m_ix_i/M$ fixed at the origin and $P=0$. 
In this case every solution $x(t)$ is in the subspace 
\begin{equation}
 \label{ecuSMomento}
 \Spo=\left\{(x_1,\dots,x_N)\in\R^N:\sum_{i=1}^N m_i x_i=0\right\}.
\end{equation}
Thus, for $M\neq 0$ we do not lose generality assuming that the center of mass if fixed at the origin.
It is remarkable that for $M=0$, the momentum of a solution does not depend on the inertial reference frame. 

\subsection{Quadratic forms}
Let us fix some notation and recall some results concerning quadratic forms that will be used. 
See \cite{HoKu} for the proofs and more information.
Let $V$ be an $n$-dimensional real vector space and let $B\colon V\times V\to\R$ be a symmetric bilinear form. 
Consider the quadratic form $Q\colon V\to\R$ defined as $Q(x)=B(x,x)$. 
We say that $B$ is \emph{degenerate} if there is $x\in V$, $x\neq 0$, such that $B(x,y)=0$ for all $y\in V$.
A basis $\{x_1,\dots,x_N\}$ of $V$ is $B$-\emph{orthogonal} if $B(x_i,x_j)=0$ for all $i\neq j$. 

If $B$ is non-degenerate then there are non-negative 
integers $p,q$ such that $n=p+q$ and
for every $B$-orthogonal basis $\{x_1,\dots,x_n\}$ of $V$ it holds that
$p=\card\{i\in\{1,\dots,N\}:Q(x_i)>0\}$ and
$q=\card\{i\in\{1,\dots,N\}:Q(x_i)<0\}$, where $\card$ denotes the cardinality of the set.
The \emph{signature} of $B$ is the ordered pair $(p,q)$.
If $(p,q)=(N,0)$ (resp. $(0,N)$) we say that $B$ is \emph{positive (resp. negative) definite}.
In any of these cases we say that $B$ is \emph{definite}, otherwise it is \emph{indefinite}. 

The Cauchy-Schwarz inequality is well known for positive (and negative) forms. 
We give a reference of the reversed inequality for indefinite forms.

\begin{lem}[\cite{Sch97}*{p. 185}]
\label{lemCS}
 If $B$ is a symmetric bilinear form which is indefinite in a plane $V$ then 
 $[B(x,y)]^ 2\geq Q(x)Q(y)$ for all $x,y\in V$.
\end{lem}

\subsection{Kinetic energy}
\label{secKinEn}
Given the masses $m_1,\dots,m_N$ define $B\colon \R^N\times \R^N\to \R$ 
as 
$$B(y,z)=\sum_{i=1}^N y_im_iz_i$$
where $y=(y_1,\dots,y_N)$ and $z=(z_1,\dots,z_N)$. 
We consider the quadratic form $Q(y)=B(y,y)$.
The kinetic energy is defined as 
$$E=\sum_{i=1}^ Nm_i\dfrac{v_i^ 2}2=\frac12Q(v).$$ 
Let $\{e_1,\dots,e_N\}$ be the standard basis of $\R^ N$. 
Notice that it is $B$-orthogonal and that the signature of $B$ is $(p,q)$ 
where $p,q$ are number of positive and negative masses respectively.
Define 
$e_{1N}=e_1+\dots+e_N=(1,\dots,1)$.
The restriction of $B$ to vectors of $\Spo$ (defined in \eqref{ecuSMomento}) will be denoted as $B|_\Spo$. 
\begin{prop}
\label{propSignPosta}
 For a system of $N\geq 3$ particles,
 if there are $p$ positive masses and $q$ negative masses then 
 \begin{itemize}
  \item if $M=0$ then $B|_\Spo$ is degenerate,
  \item if $M>0$ then the signature of $B|_\Spo$ is $(p-1,q)$,
  \item if $M<0$ then the signature of $B|_\Spo$ is $(p,q-1)$.
 \end{itemize}
\end{prop}

\begin{proof}
Notice that $\Spo=\{x\in\R^N:B(x,e_{1N})=0\}$, that is, $\Spo$ is the $B$-orthogonal complement of the vector $e_{1N}$. 
If $M=0$ then $e_{1N}\in \Spo$ and we have that $B|_\Spo$ is degenerate.
Let $y_1,\dots,y_{N-1}$ be a $B$-orthogonal basis of $\Spo$.
Since the signature of $B$ does not depend on the basis, 
if $Q(e_{1N})=M>0$ then the signature of $B|_\Spo$ is
$(j-1,k)$. 
The case $M<0$ is analogous.
\end{proof}

A fundamental property of systems of positive mass is that the energy is positive definite.
The next result implies that this property is not always lost under the presence of masses of different sign, 
at least for the restriction to $\Spo$, which is the relevant case if $M\neq 0$.

\begin{cor}
\label{corSgBrS1}
For a system of $N\geq 3$ particles on the line, 
$B|_\Spo$  
is definite 
if and only if
one of the following conditions hold: 
 \begin{itemize}
  \item all the masses have the same sign,
  \item there is exactly one mass with the sign of $M$.
 \end{itemize}
\end{cor}

We state the next result that will be applied in the proof of Theorem \ref{thmCritGralColapso}.

\begin{cor}
\label{corSgBrS2}
For a system of $N\geq 3$ particles on the line with $M\neq 0$, 
$B|_\Spo$ 
has signature $(1,N-2)$ 
if and only if one of the following conditions hold:
 \begin{itemize}
  \item $M>0$ and there exactly 2 positive masses,
  \item $M<0$ and there is just one positive mass.
 \end{itemize}
\end{cor}

The proofs of Corollaries \ref{corSgBrS1} and \ref{corSgBrS2} are direct from 
Proposition \ref{propSignPosta}.
The next result explains the Glashow-Mittag condition \eqref{ecuGM} mentioned above.

\begin{cor}
 Three particles satisfy $m_1m_2m_3(m_1+m_2+m_3)>0$ if and only if $B|_\Spo$ is definite. 
\end{cor}

\begin{proof}
The result is trivial for three masses of the same sign. 
If only one mass is negative, then the sum of the three has to be negative, 
and the result follows from Corollary \ref{corSgBrS1}. 
The case of two negative masses is analogous. 
\end{proof}

\subsection{Elastic collisions}
\label{secCollisions}

Between collision times the particles move with constant velocity. 
To consider the collision of the particle $i$ with $j$ 
define 
$$\mu_{ij}=\frac{2m_j}{m_i+m_j}.$$
If their velocities are $v_i$ and $v_j$ before the collision 
then the velocities after the collision are $w_i,w_j$ given by
\begin{equation}
\label{leychoque}
\left\{
\begin{array}{l}
w_i=(1-\mu_{ij})v_i +\mu_{ij}v_j\\
w_j=\mu_{ji}v_i + (1-\mu_{ji})v_j
\end{array}
\right.
\end{equation}
or equivalently 
\[
\left\{
\begin{array}{l}
w_i-v_i=\mu_{ij}(v_j-v_i)\\ 
w_j-v_j=\mu_{ji}(v_i-v_j).\\
\end{array}
\right.
\]
These equations are derived from the conservation laws of momentum and kinetic energy.
\begin{rmk}
If $m_i+m_j=0$ then $\mu_{ij}=\infty$ and the velocities after a collision diverge. 
Thus, we will assume that $m_i+m_j\neq 0$ for every pair of particles.  
\end{rmk}

Let us illustrate a collision between particles of masses of different sign. 
If $x_i\leq x_j$ then $v_i>v_j$, if we assume that there is a collision. 
Thus, $v_j-v_i<0$ and if $\mu_{ij}<0$ then $w_i>v_i$, i.e., the $i$-particle is accelerated. 
An illustration is given in Figure \ref{figReboteLoco}. 
\begin{figure}[ht]
\includegraphics{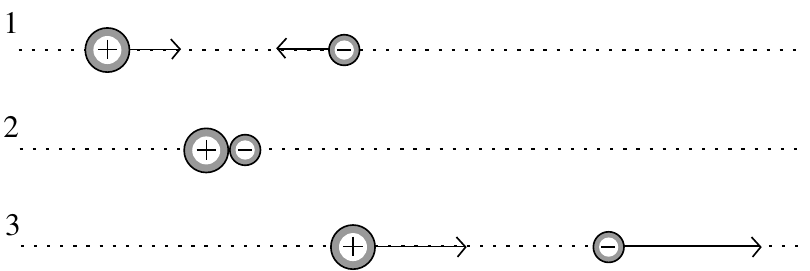}
\caption{
The circles represent the particles and the sign of the mass is indicated inside the circle.
The illustration shows a collision with $\mu_{12}<0$. 
For example, $m_1>0$ (left particle) and $m_2<0$ (right particle) with $|m_2|<m_1$. 
The numbers on the left indicate the evolution of time.}
\label{figReboteLoco}
\end{figure}
\begin{rmk}
In author's imagination the negative masses are small (in absolute value). 
This is because negative mass has not been observed in reality, thus, if it exists, it should be small. 
In \S\ref{secGravity} this idea is explored.
However, the collision illustrated in Figure \ref{figReboteLoco} also occurs if $m_1,m_2$ are replaced by $-m_1,-m_2$. 
In fact, from the kinematics we cannot distinguish a system $m_1,\dots,m_n$ from $am_1,\dots,am_N$ for $a\in\R$, $a\neq 0$.
What we observe in a collision is the coefficient $\mu_{ij}$.
\end{rmk}

\subsection{Maximal solutions}
\label{secMaxSol} 
Suppose $N$ particles in $\R$ with masses $m_1,\dots,m_N$, initial position $x(0)$ and initial velocity $v(0)$. 
The evolution of the system is described by the positions $x(t)=(x_1(t),\dots,x_N(t))$. 
The function $x(t)$ will be called as a \emph{solution}. 
Given the initial conditions, the system evolves respecting the collision law (\S \ref{secCollisions}) 
and having constant velocity between collisions. 
Solutions can be defined for all $t\in\R$ or not. 
For example, in a triple collision there may not be a satisfactory continuation of the solution. 
If at a certain time there two or more independent collisions, each collision involving only two particles, 
then the solution resolves each collision and continues. 

As we will see in \S\ref{secColSys}, 
in certain configurations of masses of different sign, there can be a convergent sequence of times $t_n$ where 
infinitely many collisions occur. 
If $t_n\to t_*$ with $t_n$ increasing or decreasing and $t_*$ finite, then we say that 
the solution has a \emph{collapse} at time $t_*$.
We say that three or more particles $i_1,\dots,i_n$ have a \emph{multiple collision} at time $t_*$ if 
$\diam\{x_{i_1}(t),\dots,x_{i_n}(t)\}\to 0$ as $t\to t_*$.
A multiple collision is \emph{direct} if there are no collisions between these particles 
in a time interval $(t_*-\epsilon,t_*)$, for some $\epsilon>0$. 
Otherwise we say that the multiple collision is \emph{indirect}. 
In general we assume that solutions have not direct multiple collisions.

\begin{prop}
 At time $t_*$ there is an indirect multiple collision or the velocity of a particle is not bounded.
\end{prop}

\begin{proof}
Suppose that the velocity of all the particles is bounded as $t$ approaches $t_*$. 
We will show that there is a triple collision. 
Let $0<t_1<t_2<\dots<t_*$ be as before.
Since the number of particles is finite, there is a pair of adjacent particles $i,j$ with infinitely many collisions. 
This implies that at least one of these particles has infinitely many collisions with the other neighbor. 
Thus, we can suppose that $i,j,k$ are three adjacent particles, $x_i\leq x_j\leq x_k$, with infinitely 
many collisions $i,j$ and $j,k$.

As we are assuming that the velocities are bounded, take $\nu>0$ such that $|v_l(t)|<\nu$ for $l=i,j,k$
and all $0\leq t<t_*$.
Suppose that the $i,j$ collisions occur at times $s_n\to t_*$.
If $s_n\leq t\leq s_{n+1}$ we have $|x_i(t)-x_j(t)|<\nu (s_{n+1}-s_n)$. 
Since $s_n\to t_*<\infty$, this implies that $|x_i(t)-x_j(t)|\to 0$ as $t\to t_*$. 
Analogously, $|x_j(t)-x_k(t)|\to 0$ as $t\to t_*$. This proves that there is an indirect multiple collision of the 
particles $i,j,k$ at $t_*$.
\end{proof}

In all what follows we will assume that the solutions are defined in a maximal interval of time.

\section{Geometric model}
\label{secGeomModel}

In this section we will show that the dynamics of $N$ particles on the line is a billiard. 
This result is well known for positive masses.
In the proofs known to the author, for instance \cite{Si73}*{Lecture 10}, 
it is shown a conjugacy, i.e. there is a map transforming the system of particles into a billiard. 
This map is usually defined by taking the square roots of the values of the masses. 
Our approach does not need these square roots, moreover, 
we do not need any transformation, instead we change the geometry of $\R^N$ by considering the form $B$ 
as an inner product.

In \S\ref{secMirror} we define the mirror law for inner products of arbitrary signature. 
In \S\ref{secEqBilPart} we show the billiard structure of systems of particles. 
In \S\ref{secMomIne} some of the geometric ideas are used to derive some properties of the moment of inertia. 

\subsection{Mirror law} 
\label{secMirror}
Billiard trajectories are a combination of linear (or geodesic) motions and reflections. 
In the classical theory of billiards, as for example in \cite{ChMa}, 
reflections are induced by a positive definite inner product. 
If in $\R^N$ we consider a non-degenerate simmetric bilinear form $B$ the mirror law can be generalized as follows \cites{Ta,KT}. 
Suppose that $S_i\subset \R^N$ is a codimension-one subspace (a wall) and a trajectory 
hits $S_i$ with velocity $v$. 
If there are vectors $l\in S_i$ and $n\in\R^N$ such that $B(l,n)=0$ and $v=l+n$, then the outgoing direction is $w=l-n$.
See Figure \ref{figRebProy}.

\begin{figure}[ht]
 \includegraphics{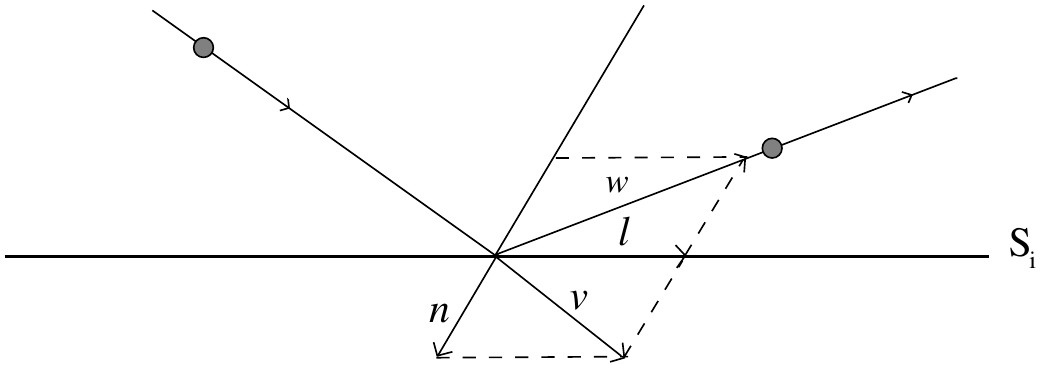}
 \caption{The mirror law associated to a form $B$.}
 \label{figRebProy}
\end{figure}

As an example, consider in $\R^2$ the inner product 
$B(a,b)=a_1a_2-b_1b_2$ where $a=(a_1,a_2)$ and $b=(b_1,b_2)$.
The associated quadratic form is $Q(a_1,a_2)=a_1^2-a_2^2$.
Suppose that $S_1$ and $S_2$ are contained in the positive cone. 
Note that $(a_2,a_1)$ is $B$-orthogonal to $(a_1,a_2)$. 
In Figure \ref{figGravitonGeom} we illustrate a trajectory of this billiard.

\begin{figure}[h]
 \includegraphics{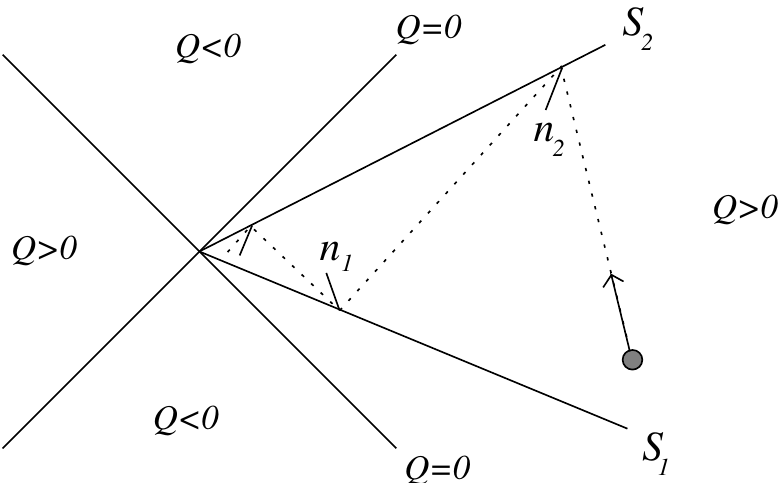}
 \caption{A trajectory in the plane with an indefinite inner product 
 that converges to the origin.}
 \label{figGravitonGeom}
\end{figure}

This kind of dynamics is related to the toy gravitons that we will see in \S\ref{secColSys}.

\subsection{Equivalence billiard-particles}
\label{secEqBilPart}
In this section we will show that the one-dimensional dynamics of a system of $N$ particles is equivalent to 
a billiard in a polyhedral angle. For this purpose it is convenient to think of $B$ (defined in \S\ref{secKinEn}) 
as an inner product. 
We start defining the geometrical notions that we will need. 
For $i,j=1,\dots,N$, $i\neq j$ 
define 
\begin{equation}
 \label{ecusGeom}
 \begin{array}{l}
 n_{ij}=\frac 1{m_j}e_j-\frac 1{m_i}e_i\\
S_{ij}=\{x\in \R^N:B(x,n_{ij})=0\}.
\end{array}
\end{equation}

Note that $B(x,n_{ij})=x_j-x_i$, thus $x\in S_{ij}$ if and only if $x_i=x_j$. 
Thus, the subspace $S_{ij}$ represents a collision between the particles $i$ and $j$.
In the next list we collect some direct formulas:
\begin{equation}
\label{ecusB}
 \begin{array}{l}
\displaystyle Q(n_{ij})=\frac{m_i+m_j}{m_im_j},\\
B(n_{ij},n_{jk})=-\frac 1{m_j},\\
B(n_{ij},n_{kl})=0,\\
\displaystyle\frac{[B(n_{ij},n_{jk})]^2}{B(n_{ij},n_{ij})B(n_{jk},n_{jk})}=\frac{m_i}{m_i+m_j}\frac{m_k}{m_j+m_k}=\mu_{ji}\mu_{jk}
\end{array}
\end{equation}
for different $i,j,k,l$.

From \eqref{leychoque} we have that the 
linear map
$T_{ij}\colon \R^N\to\R^N$ defined as 
\begin{equation}
\label{ecuTij}
 \left\{
 \begin{array}{l}
T_{ij}(e_i) = (1-\mu_{ij})e_i+\mu_{ji}e_j,\\
T_{ij}(e_j) = \mu_{ij}e_i+(1-\mu_{ji})e_j,\\
T_{ij}(e_k)=e_k, \text{ for }k\neq i, k\neq j
 \end{array}
 \right.
\end{equation}
represents a collision between the $i$ and $j$ particles. 

\begin{prop}
\label{propBilMass}
The form $B$ is preserved by $T_{ij}$ and it holds that $T_{ij}(n_{ij})=-n_{ij}$ and $T_{ij}(v)=v$ for all $v\in S_{ij}$.
Consequently, the dynamics of $N$ particles on the line 
is a billiard in $\R^N$
with walls $S_{ij}$ and normals $n_{ij}$ respectively.
\end{prop}

\begin{proof}
From the definitions and the linearity of $T_{ij}$ we have
\[
\begin{array}{ll}
T_{ij}(n_{ij})& =T(\frac 1{m_j}e_j-\frac 1{m_i}e_i)=\frac 1{m_j}T(e_j)-\frac 1{m_i}T(e_i)\\
& =\frac 1{m_j}(\mu_{ij}e_i+(1-\mu_{ji})e_j)-\frac 1{m_i}((1-\mu_{ij})e_i+\mu_{ji}e_j)\\
& =\frac 1{m_j}\mu_{ij}e_i+\frac 1{m_j}(1-\mu_{ji})e_j-\frac 1{m_i}(1-\mu_{ij})e_i-\frac 1{m_i}\mu_{ji}e_j\\
& =[\frac 1{m_j}\mu_{ij}-\frac 1{m_i}(1-\mu_{ij})]e_i+[\frac 1{m_j}(1-\mu_{ji}) -\frac 1{m_i}\mu_{ji}]e_j\\
& =[\frac 2{m_i+m_j}-\frac {m_i-m_j}{m_i(m_i+m_j)}]e_i+[\frac {m_j-m_i}{m_j(m_i+m_j)} -\frac 2{m_i+m_j}]e_j\\
& =\frac 1{m_i}e_i-\frac 1{m_j}e_j=-n_{ij}\\
\end{array}
\]
Also
\[
 \begin{array}{ll}
  T_{ij}(e_i+e_j)& =T_{ij}(e_i)+T_{ij}(e_j)= 
  (1-\mu_{ij})e_i+\mu_{ji}e_j + \mu_{ij}e_i+(1-\mu_{ji})e_j\\
  &=e_i+e_j
 \end{array}
\]
Since $v=\sum_{l=1}^Nv_le_l$ is in $S_{ij}$ if and only if $v_i=v_j$, we conclude that 
$T_{ij}(v)=v$ for all $v\in S_{ij}$.
This proves that the dynamics is a billiard.
\end{proof}

The next result remarks that certain configurations of masses of different sign have dynamics with finitely many collisions.
It extends \cite{Ga78} where it is proved that a system of positive masses has finitely many collisions.

\begin{cor}
\label{corFinColUno}
If a system of $N$ particles on the line 
has all the masses with the same sign or
there is just one mass with the sign of $M$
then every solution has finitely many collisions. 
\end{cor}

\begin{proof}
The result for positive masses was proved in \cite{Ga78}. 
If all the masses are negative then a change of sign reduces the situation to the positive case. 
If there is just one mass with the sign of $M$, then by Corollary \ref{corSgBrS1} we know that $B|_\Spo$ is definite. 
We can assume that it is positive definite.
By Proposition \ref{propBilMass} we have that the system is equivalent 
to a billiard in a polyhedral angle. 
Therefore, by \cite{Si78} we know that the trajectories in $\Spo$ have finitely many collisions. 
Recall that the solutions contained in $\Spo$ are those with the center of mass fixed at the origin. 
Finally, notice that a change of inertial frame does not change the number of collisions of a solution.
\end{proof}

%

\begin{rmk}
Let us explain why the result of \cite{Si78} is strictly stronger than \cite{Ga78}, 
where it is shown that every trajectory has finitely many collisions on a polyhedral angle with respect to a positive inner product.
On one hand, there are configurations of particles with negative mass 
giving rise to a positive $B|_\Spo$ (i.e. a standar polyhedral angle).
On the other hand, for a system of 4 particles, we see from
\eqref{ecusB} that $B(n_{12},n_{34})=0$. This is a restriction to the possible polyhedral angles 
that can be obtained from systems of particles on a line. 
\end{rmk}

\subsection{Moment of inertia} 
\label{secMomIne}
The moment of inertia of a configuration $x\in\R^N$ is 
$I(x)=B(x,x)=\sum_{i=1}^N m_i x_i^ 2$. 
We will describe the evolution of this number along a trajectory. 
Suppose a solution $x(t)$ with 
 initial conditions $x(0),v(0)$ and collisions at times $t_0=0<t_1<t_2<\dots$, where at $t_n$ there is a collision between 
 the particles $i_n$ and $j_n$. 
 Let $T_n=T_{i_n,j_n}$, where $T_{i,j}$ is given by \eqref{ecuTij}.
 Let $\tilde v_n\in\R^N$, for $n\geq 0$, be defined by $\tilde v_0=v(0)$ and $\tilde v_n=T_n(\tilde v_{n-1})$.
For $t_n\leq t\leq t_{n+1}$ we have that 
$x(t)= x(t_n)+(t-t_n)\tilde v_n$.
\begin{lem}
\label{lemaDesarrollo}
 For $t_n\leq t\leq t_{n+1}$ it holds that 
 \[
  T_1\circ\dots\circ T_n(x(t))=x(0)+tv(0).
 \]
\end{lem}
\begin{proof}
To prove it by induction, notice that for $n=0$ we have 
$x(t)=x(0)+tv(0)$ for all $t\in[0,t_1]$. 
Now assume that the statement holds for $t_{n-1}\leq t\leq t_n$. 
This implies that 
$T_1\circ\dots\circ T_{n-1}(x(t_n))=x(0)+t_nv(0)$. 
Since at time $t_n$ there is a collision between the particles $i_n$ and $j_n$ we have that 
$x(t_n)\in S_{i_n,j_n}$ (the subspaces defined in \eqref{ecusGeom}).
Therefore, Proposition \ref{propBilMass} implies that $T_n(x(t_n))=x(t_n)$ and 
$$T_1\circ\dots\circ T_{n-1}\circ T_n(x(t_n))=x(0)+t_nv(0).$$
For $t\in[t_n,t_{n+1}]$ we know that $x(t)=x(t_n)+(t-t_n)\hat v_n$ and 
\[
 \begin{array}{ll}
  T_1\circ\dots\circ T_n(x(t))
  & = T_1\circ\dots\circ T_n(x(t_n)+(t-t_n)\hat v_n)\\
  & = T_1\circ\dots\circ T_n(x(t_n))+(t-t_n)T_1\circ\dots\circ T_n(\hat v_n)\\
  & = x(0)+t_nv(0)+(t-t_n)\hat v_0\\
  &=x(0)+tv(0).
 \end{array}
\]
This finishes the proof.
\end{proof}

\begin{prop}
For every solution $x(t)$ it holds that
\begin{equation}
 \label{ecuEvoB}
 I(x(t))=I(x(0))+2tB(x(0),v(0))+2t^2E=Q(x(0)+tv(0))
\end{equation}
\end{prop}
\begin{proof}
The second equality follows from the bilinearity of $B$.
By Proposition \ref{propBilMass} we know that each $T_n$ is a $B$-isometry. 
Thus, $Q(T_1\circ\dots\circ T_n(x(t)))=Q(x(t))$. 
Therefore, the result follows by Lemma \ref{lemaDesarrollo}.
\end{proof}

%

\section{Collapsing systems}
\label{secColSys}

In this section we give a characterization of systems with collapsing solutions.

\subsection{Global collapse}

%
%
%

We say that a solution $x(t)$ has a \emph{global collapse} at time $s\in\R$ if 
the system collapses at time $s$ (according to \S\ref{secMaxSol}) and each pair of adjacent particles has infinitely many collision as $t\to s$.
We say that it has a \emph{double global collapse} if it has a global collapse at $s_1>0$ and at $s_2<0$.

Given $1\leq j\leq k\leq N-1$ define 
$$m_{j,k}=m_j+\dots+m_k.$$

\begin{lem}
\label{lemBPosEnCono}
If a system satisfies
\begin{equation}
\label{ecuBPosEnCono}
m_{1,j}m_{k+1,N}M>0\text{ for all }1\leq j \leq k\leq N-1,
\end{equation}
then for all $x\in \Spo$ such that $x_1\leq\dots\leq x_N$ it holds that
$I(x)\geq 0$ with equality only if $x=0$.
\end{lem}

\begin{proof}
Given $1\leq k\leq N-1$ define 
$e_{1k}=e_1+\dots+e_k$
and 
%
\[
 \xi_k=\frac{m_{1,k}}{M}e_{1N}-e_{1k}.  
\]
If $j\leq k$ then $B(e_{1j},e_{1k})=m_{1,j}$ and
\[ 
\begin{array}{ll}
 B(\xi_j,\xi_k) & =B(\frac{m_{1,j}}{M}e_{1N}-e_{1j},\frac{m_{1,k}}{M}e_{1N}-e_{1k})\\
 & =M^{-2}B(m_{1,j}e_{1N}-Me_{1j},m_{1,k}e_{1N}-Me_{1k})\\
 & =M^{-2}(
 m_{1,j}m_{1k}M
 -m_{1,j}M m_{1,k}
 -Mm_{1,k} m_{1,j}+
 M^2 m_{1,j}
 )\\
 & =M^{-1}(
 -m_{1,k} m_{1,j}+
 M m_{1,j}
 ) =M^{-1}m_{1,j}(-m_{1,k} + M)\\
 &= m_{1,j}m_{k+1,N}M^{-1}
\end{array}
\]
for all $1\leq j\leq k\leq N-1$. 
Thus, \eqref{ecuBPosEnCono}
is equivalent to 
$B(\xi_j,\xi_k)>0$ for all $1\leq j\leq k\leq N-1$.

Given $x=\sum_{k=1}^N x_ke_k\in \Spo$ (i.e., $\sum_{k=1}^Nm_kx_k=0$) we have
\[
 \begin{array}{lll}
  \sum_{k=1}^{N-1}(x_{k+1}-x_k)\xi_k 
   & = & -x_1\xi_1+\sum_{k=2}^{N-1} x_k(\xi_{k-1}-\xi_k) + x_N\xi_{N-1}\\
   & = & x_1(e_1-\frac{m_1}{M}e_{1N})\\
   &&+\sum_{k=2}^{N-1} x_k(e_k-\frac{m_k}Me_{1N}) + x_N(e_N-\frac{m_N}{M}e_{1N})\\
   & = & \sum_{k=1}^N x_ke_k-\sum_{k=1}^N \frac{m_kx_k}{M}e_{1N}=x.
 \end{array}
\]
This implies that $I(x)>0$ for all $x\neq 0$ such that 
$x_1\leq\dots\leq x_N$. 
%
%
%
\end{proof}

\begin{rmk}
From a geometric viewpoint, the condition $I(x)\geq 0$ of Lemma \ref{lemBPosEnCono} means 
that the polyhedral angle given by $x_1\leq\dots\leq x_N$ is contained in a positive cone as in 
the example shown in Figure \ref{figGravitonGeom}.
\end{rmk}

In what follows we will consider systems satisfying the next conditions: 
\begin{equation}
 \label{ecuGasNeg}
 \left\{
 \begin{array}{l}
  m_2,\dots,m_{N-1}<0,\\
  m_1+m_2+\dots+m_{N-1}>0,\\
  (m_2+\dots+m_N)M>0,\\
  x_1\leq\dots\leq x_N.
 \end{array}
 \right.
\end{equation}

\begin{rmk}
 Condition \eqref{ecuGasNeg} implies \eqref{ecuBPosEnCono}. To prove it, first note that \eqref{ecuGasNeg} implies
 $m_1+\dots+m_j>0$ for all $1\leq j \leq N-1$.
Therefore, to prove \eqref{ecuBPosEnCono} we have to show that 
 $(m_{k+1}+\dots+m_N)M>0$ for all $1 \leq k\leq N-1$.
If $M>0$, then $m_2+\dots+m_N>0$, and since $m_{k+1}+\dots+m_N\geq m_2+\dots+m_N$ we conclude \eqref{ecuBPosEnCono}.
If $M<0$, since $m_1+m_2+\dots+m_{N-1}>0$ we have that $m_N$ is negative (and large in absolute value). 
Then, $m_{k+1}+\dots+m_N<0$ for all $1 \leq k\leq N-1$, which proves \eqref{ecuBPosEnCono}.
 
\end{rmk}

\begin{rmk}
A system of 3 particles satisfies \eqref{ecuGasNeg} if and only if: 
\begin{itemize}
 \item (toy graviton \S\ref{secToyG}) $m_1+m_2>0$, $m_2<0$ and $m_2+m_3>0$ or
 \item (compressor \S\ref{secCompr}) $m_1+m_2>0$, $m_2<0$ and $M<0$.
\end{itemize}
\end{rmk}

We say that a solution is \emph{trivial} if $v_1=\dots =v_N$. 
As we will remark below, the next result applies to toy gravitons and compressors.

\begin{thm}
\label{thmCritGralColapso}
If a system satisfies \eqref{ecuGasNeg} 
then 
every collapse is global, 
$B|_\Spo$ has signature $(1,N-2)$ 
and for every non-trivial solution $x(t)$ without direct multiple collisions and with energy $E$ it holds that:
\begin{itemize}
 \item If $E<0$ then the solution has a double global collapse.
 \item If $E\geq 0$ then $x(t)$ has a global collapse.
\end{itemize}
\end{thm}

\begin{proof}
Let $s$ be the time of a collapse.
Arguing by contradiction suppose that the particles $i,i+1$ have no collision in a neighborhood of $s$.
This implies that at least one of the subsystems $A=\{x_1,\dots,x_i\}$ or $B=\{x_{i+1},\dots,x_N\}$ has a collapse. 
In the subsystem $A$, by \eqref{ecuGasNeg} only $m_1$ (which is positive) has the sign of $m_1+\dots+m_i$. 
Thus, Corollary \ref{corFinColUno} implies that $A$ cannot collapse.
If $M<0$ then \eqref{ecuGasNeg} implies that $m_N<0$.
Thus, the masses of the subsystem $B$ are all negative, and by Corollary \ref{corFinColUno} $B$ cannot collapse.
Suppose that $M>0$. 
Applying \eqref{ecuBPosEnCono} with $j=k=N-1$ we obtain 
$(m_1+\dots+m_{N-1})m_NM>0$. Since $m_1+\dots+m_{N-1}>0$ we have that $m_N>0$. 
Again, \eqref{ecuBPosEnCono} applied for $j=k=i$ gives 
$(m_1+\dots+m_i)(m_{i+1}+\dots+m_N)M>0$. 
This implies that $m_{i+1}+\dots+m_N>0$. 
Considering the subsystem $B$, 
only $m_N$ has the sign of $m_{i+1}+\dots+m_N$ (positive), thus 
Corollary \ref{corFinColUno} implies that $B$ cannot collapse. 
This proves that every collapse is global.

The fact that $B|_\Spo$ has signature $(1,N-2)$ follows directly from Corollary \ref{corSgBrS2}, jointly with the previous considerations. 

Suppose that $x(t)$ is a solution without direct multiple collisions and with energy $E$.
By Lemma \ref{lemBPosEnCono} we know that $I(x)>0$
whenever $x_1\leq\dots\leq x_N$ and $x\neq 0$.
Consider 
\begin{equation}
 \label{ecuPoly}
 p(t)=I(x(0))+2tB(x(0),v(0))+2t^2E.
\end{equation}
If $E<0$ then $p(t)$ has two roots $s_1<0<s_2$.
By Equation \eqref{ecuEvoB} we have that the solution cannot be defined beyond the interval $(s_1,s_2)$. 
Thus, there is a double collapse (which as we proved, has to be global).
%

If $E=0$ then $p(t)=I(x(0))+2tB(x(0),v(0))$. 
We will show that $B(x(0),v(0))\neq 0$. 
Since the signature of $B|_\Spo$ is $(1,N-2)$ and $I(x(0))>0$ we 
have that $B|_{\Spo\cap x(0)^ \perp}$ is negative definite, where $x(0)^\perp$ denotes the $B$-orthogonal complement of $x(0)$. 
As $v(0)\neq 0$ and $E=Q(v(0))=0$, we have that $v(0)\notin x(0)^\perp$ 
and $B(x(0),v(0))\neq 0$.
Thus, $p(t)$ has a root at $s=-I(x(0))/2B(x(0),v(0))$ and the previous arguments implies that there is a global collapse.

Finally, assume that $E>0$. 
Since the solution has not direct multiple collisions we have that $x(0)$ and $v(0)$ are 
linearly independent. Let $V\subset \Spo$ be the plane generated by these vectors.
Since $B|_\Spo$ has signature $(1,N-2)$ we conclude that $B|_V$ is indefinite with signature $(1,1)$. 
By Lemma \ref{lemCS} we have that 
$[B(x(0),v(0))]^2\geq I(x(0))2E$. 
This implies that $p(t)$ always has at least one root 
(and if it has two roots then they have the same sign). 
Thus, the result is proved.
\end{proof}

Notice that for $N$ particles as in Theorem \ref{thmCritGralColapso} we have not proved that the solutions 
are defined for all $t\in (s_1,s_2)$. 
In the next result we prove it for $N=3$, that is the case we need for \S\ref{secGravity}.
It would be interesting to give a proof for general $N$. 

\begin{thm}
\label{thmCritGralColapsoDos}
Suppose that 
$x(t)$ is a solution with $E<0$ of a 
system of $3$ particles that satisfies \eqref{ecuGasNeg}.
Then, $x(t)$ is defined for all $t\in (s_1,s_2)$, where $s_1,s_2$ are roots of the polynomial 
\eqref{ecuPoly}. 
In particular, $x_3(t)-x_1(t)\to 0$ as $t\to s_i$, $i=1,2$.
\end{thm}

\begin{proof}
By Theorem \ref{thmCritGralColapso} we know that the signature of $B|_\Spo$ is $(1,1)$. 
Let $T=T_{12}\circ T_{23}|_\Spo$. 
From Proposition \ref{propBilMass} we see that $T_{12}|_\Spo$ and $T_{23}|_\Spo$ are reflections with determinant $-1$. 
Thus $\det(T)=1$. 
This implies that $T$ preserves the orientation of the plane $\Spo$. 
By Proposition \ref{propBilMass} we also know that $T_{12}$ and $T_{23}$ are $B$-isometries, 
which implies that $T$ preserves the positive and negative cones of $B$. Consequently, $T$ preserves each line $E=0$. 
From this, we see that $T$ is diagonalizable with real eigenvalues (two linearly independent and invariant lines are those with $E=0$).
Let us show that $T$ is not $\pm Id$ (the identity map of $\Spo$).

For three particles we have 
\[
 \left\{\begin{array}{l}
  \xi_1=\frac1M(-m_2-m_3,m_1,m_1)\\
  \xi_2=\frac1M(-m_3,-m_3,m_1+m_2)\\
  \end{array}
  \right.
  \left\{
  \begin{array}{l}
  n_{12}=(-1/m_1,1/m_2,0)\\
  n_{23}=(0,-1/m_2,1/m_3).
 \end{array}
 \right.
\]
We have that $\{\xi_1,n_{23}\}$ and $\{\xi_2,n_{12}\}$ are $B$-orthogonal basis of $\Spo$. 
If we define 
\[
A_1=\left(
\begin{array}{ll}
a & b\\
c & d
\end{array}
\right)
=
\left(
\begin{array}{ll}
m_1/(m_1+m_2) & M/m_3(m_1+m_2)\\
-m_1m_2/(m_1+m_2) & -m_1/(m_1+m_2)
\end{array}
\right)
\]
then 
\[
\left\{ 
\begin{array}{l}
  \xi_1=a\xi_2+cn_{12},\\
  n_{23}=b\xi_2+dn_{12},
 \end{array}
\right.
\]
It holds that $T$ in the basis $\{\xi_1,n_{23}\}$
is given by the matrix
\[
A_2=
\frac{1}{ad-bc}
\left(
\begin{array}{cc}
 -(ad+bc) & 2bd \\
 2ac & - (ad+bc) \\
 \end{array}
 \right).
\]
Which proves that $T\neq\pm Id$.
Thus, the eigenvalues of $T$ are $\lambda,\lambda^{-1}\in\R$, $|\lambda|\neq 1$.

%

From Lemma \ref{lemaDesarrollo} we have 
that if $t_n\leq t\leq t_{n+1}$ then
 \[
  T_1\circ\dots\circ T_n(x(t))=x(0)+tv(0).
 \]
For our particular case of 3 particles, we have that $T_1,\dots,T_n$ alternate between $T_{12}$ and $T_{23}$ (there are no more possible collisions). 
Therefore, if $n=2k$, then 
 \[
  T^k(x(t_{2k}))=x(0)+t_{2k}v(0).
 \]
Since $x(t_{2k})\in S_{12}$, we have that 
$x(0)+t_{2k}v(0)\in T^k(S_{12})$. 
If $l=\Spo\cap S_{12}$ then $x(0)+t_{2k}v(0)$ is the intersection of the lines $\{x(0)+tv(0):t\in\R\}$ and $T^k(l)$.
As $T$ is hyperbolic, we have that $T^k(l)$ converges to a direction with $E=0$. 
Thus, $Q(x(0)+t_{2k}v(0))\to 0$ as $k\to+\infty$.
This proves that $t_{2k}\to s_2$ as $k\to +\infty$. Analogously, considering the solution with initial velocity $-v(0)$ 
we have that the solution is defined for all $t\in (s_1,s_2)$. 
Finally, from Lemma \ref{lemBPosEnCono} we conclude that
$x_3(t)-x_1(t)\to 0$ as $t\to s_i$, $i=1,2$.
\end{proof}

\subsection{Toy gravitons}
\label{secToyG}
We say that $m_1,\dots,m_N$, $N\geq 3$, is a \emph{system of gravitons} if 
\begin{equation}
 \label{ecuGraviton}
 \left\{
 \begin{array}{l}
m_1,m_N>0>m_2,\dots,m_{N-1},\\
|m_2+\dots+m_{N-1}|<\min\{m_1,m_N\}. 
 \end{array}
 \right.
\end{equation}
We assume that $x_1\leq\dots\leq x_N$.

\begin{rmk}
\label{rmkToyG}
Toy gravitons satisfy \eqref{ecuGasNeg}, and consequently, Theorem \ref{thmCritGralColapso}. 
In fact, a system $m_1,\dots,m_N$ is a toy graviton if and only if it satisfies \eqref{ecuGasNeg} with $M>0$.
\end{rmk}

\begin{rmk}
The simplest example consists of three masses $m_1,m_2,m_3$ such that $m_1+m_2,m_2+m_3>0>m_2$. 
%
\end{rmk}

At each collision $1,2$ we have that the velocity of $m_1$ increases while 
at a collision $N-1,N$ the velocity of $m_N$ 
decreases (but its modulus increases, pointing to the left). 
This explains why we can have infinitely many collisions in a finite time interval.

\begin{figure}[h]
\includegraphics{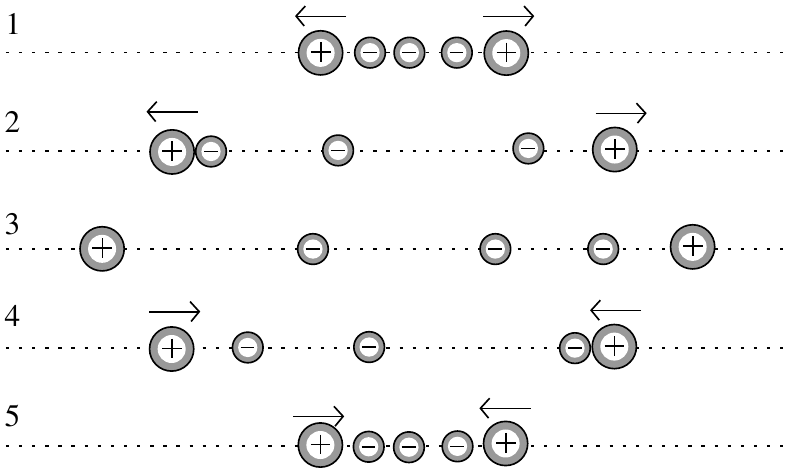}
\caption{A toy graviton system collapsing two positive masses.}
\label{figToyGraviton}
\end{figure}

\subsection{Compressors} 
\label{secCompr}
A \emph{compressor} is a system of $N\geq 3$ particles such that 
$x_1\leq\dots\leq x_N$ and 
\begin{equation}
\label{ecuRunaway}
\left\{
\begin{array}{l}
m_1>0>m_2,\dots,m_N,\\
m_1+\dots+m_{N-1}>0,\\
\sum_{i=1}^Nm_i<0.   
\end{array}
\right. 
\end{equation}

\begin{rmk}
\label{rmkRunaway}
Compressor systems satisfy \eqref{ecuGasNeg}, and consequently, Theorem \ref{thmCritGralColapso}. 
In fact, a system $m_1,\dots,m_N$ is a compressor if and only if it satisfies 
\eqref{ecuGasNeg} with $M<0$.
\end{rmk}

The next result means that the center of mass of a compressor system is at the right of $x_N$. 
We assume that the center of mass is $0$. 

\begin{prop}
\label{propCompCentMass}
For a compressor it holds that $x_i\leq 0$ for all $i=1,\dots,N$.
\end{prop}

\begin{proof}
Arguing by contradiction, suppose that $x_j> 0$ for some $j=1,\dots,N$. 
First we consider the case $j>1$ and $x_{j-1}\leq0$.
Since $m_i<0$ for all $i\geq 2$ we have that 
$\sum_{i=j}^N m_ix_i< 0$. 
From \eqref{ecuRunaway} we also have
\[
 \sum_{i=1}^{j-1}m_ix_i\leq x_1\sum_{i=1}^{j-1}m_i\leq 0
\]
which contradicts that the center of mass is $0$. 
For $j=1$ we have $x_1>0$ which gives the contradiction 
$$\sum_{i=1}^Nm_ix_i\leq x_1\sum_{i=1}^Nm_i<0.$$
This finishes the proof.
\end{proof}

The following result is a direct explanation of the collapse of a compressor.

\begin{prop}
For a compressor, if $v_1(0)>0$ then the solution collapses.
\end{prop}

\begin{proof}
As $\mu_{12}<0$, at each collision between $m_1$ and $m_2$, the particle $m_1$ is accelerated. As it starts with positive velocity, if there were only finitely many collisions, 
the particle $m_1$ would cross $0$, 
contradicting Proposition \ref{propCompCentMass}. 
Thus, the solution has infinitely many collisions in finite time. 
\end{proof}

From this proof we see that $m_1$ compresses the gas formed by the particles of negative mass.
See Figure \ref{figCompressor}.

\begin{figure}[h]
 \includegraphics{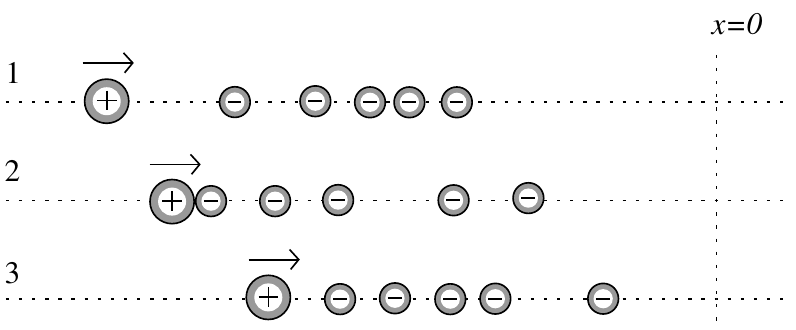}
 \caption{The gas formed by the particles of negative mass is compressed by $m_1$.}
 \label{figCompressor}
\end{figure}



%

\subsection{Systems with finitely many collisions}


In the next result \emph{every solution} means solutions with any initial condition (velocities and positions), in particular we will 
not assume that the particles start in order $x_1\leq,\dots,\leq x_N$, any initial order is allowed.

\begin{thm}
\label{thmSgBrS}
For a system of $N\geq 3$ particles on the line the following statements are equivalent: 
\begin{enumerate}
 \item[(a)] $B|_\Spo$ is definite,
 \item[(b)] $M\neq 0$ and every solution has finitely many collisions. 
\end{enumerate}
\end{thm}

\begin{proof}
By Proposition \ref{propSignPosta} we know that (a) implies that $M\neq 0$.
The fact that (a) implies finitely many collisions is a combination of Corollaries \ref{corSgBrS1} and \ref{corFinColUno}.

To prove that (b) implies (a) we will apply 
Corollary \ref{corSgBrS1}. 
Assume that $M>0$.
Arguing by contradiction, 
suppose that there are 
at least two positive masses and one negative.
Suppose that we have two particles with each sign: $m_1\leq m_2<0<m_3\leq m_4$. 
If $m_2+m_3>0$ then $m_3,m_2,m_4$ make a toy graviton.
If $m_2+m_3<0$ then $m_1,m_3,m_2$ make a toy graviton. 
By Remark \ref{rmkToyG} toy gravitons have solutions with infinitely many collisions.
Now assume that there is just one negative mass.
Suppose that $m_1<0<m_2\leq\dots\leq m_N$.
If $m_1+m_2>0$ then $m_2,m_1,m_3$ make a toy graviton. 
If $m_1+m_2<0$ then $m_1,\dots,m_N$ contains a compressor. By Remark \ref{rmkRunaway}, compressors have solutions with 
infinitely many collisions. Thus, in any case we arrived to a contradiction and the proof ends.
\end{proof}

\section{Potential energy}
\label{secGravity}
Consider a graviton system of 3 particles with the following initial conditions: 
\[
 \begin{array}{l}
  x_1(0)<x_2(0)=0< x_3(0),\\
  m_1x_1(0)+m_3x_3(0)=0,\\
  v_1(0)=v_3(0)=0.
 \end{array}
\]
We assume that $m_1,m_3>0$ are fixed and $m_2<0$ is variable.
We will take the limit $m_2\to 0$ keeping the kinetic energy of particle constant, 
$\frac 12 m_2v_2^2(0)=U_0<0$, which implies that
$v_2(0)\to\infty$. 

The kinetic energy of the system is $E=U_0$ (since $v_1(0)=v_2(0)=0$).
Let $I_0=m_1x_1^2(0)+m_3x_3^2(0)$ (recall that $x_2(0)=0$). 
From \eqref{ecuEvoB} we have
\[
 I(x(t))=I(x(0))+2tB(x(0),v(0))+2t^2E
\]
and 
\begin{equation}
 \label{ecuInerGrav}
 m_1x_1^2(t)+m_2x_2^2(t)+m_3x_3^2(t)=I_0+2t^2E
\end{equation}
because $B(x(0),v(0))=0$.
Let $\pm s$ be the roots of $I_0+2t^2E=0$.
By Theorem \ref{thmCritGralColapsoDos}, independently of the value of $m_2$ the solution $x(t)$ is defined for all 
$t\in (-s,s)$.
The center of mass starts at the origin and the momentum is 
$P=m_2v_2(0)$ (because $v_1(0)=v_3(0)=0$). Therefore 
\begin{equation}
\label{ecuCentMassGrav}
 m_1x_1(t)+m_2x_2(t)+m_3x_3(t)=tm_2v_2(0).
\end{equation}

Notice that $m_2v_2(0)=m_2v_2^2(0)/v_2(0)\to 0$ as $m_2\to 0$.
Let $r_0=x_3(0)-x_1(0)$. 
Since $m_1,m_3$ start at rest, we have that $|x_2(t)|\leq r_0$ for all $t\in(-s,s)$ and all $m_2$.
Then $m_2x_2(t)\to 0$ and 
$m_2x^2_2(t)\to 0$
as $m_2\to 0$.
Consequently if we take the limit $m_2\to 0$ in 
\eqref{ecuInerGrav} and \eqref{ecuCentMassGrav} 
we obtain 
\[
 \begin{array}{l}
  \lim_{m_2\to 0}m_1x_1^2(t)+m_3x_3^2(t)=I_0+2t^2E\\
  \lim_{m_2\to 0}m_1x_1(t)+m_3x_3(t)=0
 \end{array}
\]
for all $t\in (-s,s)$.
These limits make sense as $x_1(t),x_3(t)$ depend on $m_2$. 
Therefore, each limit $\hat x_i(t)=\lim_{m_2\to 0} x_i(t)$ 
exists
for fixed $t\in(-s,s)$, $i=1,3$ and they satisfy 
\begin{equation}
 \label{ecuInCent}
\begin{array}{l}
  m_1\hat x_1^2(t)+m_3\hat x_3^2(t)=I_0+2t^2E\\
  m_1\hat x_1(t)+m_3\hat x_3(t)=0
 \end{array}
 \end{equation}
Note that $\hat x_1(t)$ and $\hat x_3(t)$ are smooth, 
and denote as $\hat v_i(t)$ their velocities, $t\in(-s,s)$.
Let $T=\frac12m_1\hat v_1^2+\frac12m_3\hat v_3^2$ be the kinetic energy of the limit particles. 
Define $r=\hat x_3-\hat x_1$ 
and the potential energy
$U=\dfrac{U_0r_0^2}{r^2}$.

\begin{thm}
\label{thmGravLim}
For the limit system
it holds that
$$T+U=E$$ 
is a constant of motion. 
Therefore, $U$ is the limit of the kinetic energy of the toy graviton.
\end{thm}

\begin{proof}
Since $m_1\hat x_1(t)+m_3\hat x_3(t)=0$ 
we have that $\hat x_1=-m_3 r /(m_1+m_3)$ and 
$\hat x_3=m_1 r /(m_1+m_3)$. 
Defining $a=\dfrac{m_1m_3}{2(m_1+m_3)}$, 
by \eqref{ecuInCent} 
 we conclude 
$r^2=r_0^2+t^2E/a$
and 
$r\dot r=tE/a$. 
Since, $\hat v_3=-m_1\hat v_1/m_3$ we have $T=a\dot r^2$. 
Then 
\[
 \begin{array}{ll}
  T+U & = a\dot r^2+ \dfrac{U_0r_0^2}{r^2} = a\left(\dfrac{tE}{ar}\right)^2+ \dfrac{Er_0^2}{r^2}\\
 & = \dfrac{E}{r^2}\left(\dfrac{t^2 E}{a}+r_0^2\right) = \dfrac{E}{r^2}r^2=E
 \end{array}
\]
Which proves the result.
\end{proof}

This result allows us to
understand the negative potential energy $U(r)$ as the kinetic energy of the limit toy graviton. 
It would be interesting to find a modification of our model in order to obtain, 
for instance, a Newtonian gravitational potential $V(r)=-k/r$.

\end{document}